\documentclass[11pt,a4paper]{article}

\usepackage[utf8]{inputenc}
\usepackage[T1]{fontenc}
\usepackage[english]{babel}
\usepackage[pdftex]{graphicx}

\usepackage{amsmath,amssymb,amsfonts,amsthm}
\newtheorem{lemma}{Lemma}
\newtheorem{theorem}{Theorem}

\theoremstyle{definition}
\newtheorem{definition}{Definition}

\usepackage{fullpage}

\usepackage{cleveref}
\crefname{@theorem}{Theorem}{Theorems}
\newtheorem{Observation}{Observation}[section]

\newcommand{\calR}{\mathcal R}

\newcommand{\ceil}[1]{\left\lceil{#1}\right\rceil}
\newcommand{\floor}[1]{\left\lfloor{#1}\right\rfloor}
\newcommand{\eps}{\varepsilon}

\newcommand{\xor}{\oplus}

\newcommand{\sm}{\,\setminus\,}
\newcommand{\abs}[1]{\left | #1 \right |}
\newcommand{\set}[1]{\left \{ #1 \right \}}
\newcommand{\oneToN}[1]{\set{1, \ldots, #1}}

\newcommand*\samethanks[1][\value{footnote}]{\footnotemark[#1]}

\def\squarebox#1{\hbox to #1{\hfill\vbox to #1{\vfill}}}
\def\qed{\hspace*{\fill}
    \vbox{\hrule\hbox{\vrule\squarebox{.5em}\vrule}\hrule}}

\usepackage[dvipsnames,usenames]{xcolor}

\begin{document}

\title{The Power of Two Choices with Simple Tabulation}
\author{Søren Dahlgaard\thanks{Research partly supported by Mikkel Thorup's
    Advanced Grant from the Danish Council for Independent Research under the
Sapere Aude research carrier programme.}, Mathias Bæk Tejs Knudsen\samethanks\
\thanks{Research partly supported by the FNU project AlgoDisc - Discrete
Mathematics, Algorithms, and Data Structures.}, Eva Rotenberg, and Mikkel
Thorup\samethanks[1] \\ University of Copenhagen \\
\tt{\{soerend,knudsen,roden,mthorup\}@di.ku.dk}
}
\date{}

\maketitle


\begin{abstract}
The power of two choices is a classic paradigm for load balancing when
assigning $m$ balls to $n$ bins. When placing a ball, we pick two bins
according to two hash functions $h_0$ and $h_1$, and place the ball
in the least loaded bin.  Assuming fully random hash functions, when
$m=O(n)$, Azar et al.~[STOC'94] proved that the maximum load is $\lg
\lg n + O(1)$ with high probability.

In this paper, we investigate the power of two choices when the hash
functions $h_0$ and $h_1$ are implemented with simple tabulation,
which is a very efficient hash function evaluated in constant
time. Following their analysis of Cuckoo hashing [J.ACM'12],
P\v{a}tra\c{s}cu and Thorup claimed that the expected maximum load
with simple tabulation is $O(\lg\lg n)$. This did not include any high
probability guarantee, so the load balancing was not yet to be
trusted.

Here, we show that with simple tabulation, the maximum load is
$O(\lg\lg n)$ with high probability, giving the first constant time hash
function with this guarantee. We also give a concrete example
where, unlike with fully random hashing, the maximum load is not
bounded by $\lg \lg n + O(1)$, or even $(1+o(1))\lg \lg n$ with high
probability. Finally, we show that the expected maximum load is $\lg
\lg n + O(1)$, just like with fully random hashing.
\end{abstract}

\section{Introduction}
Consider the problem of placing $n$ balls into $n$ bins. If the balls
are placed independently and uniformly at random, it is well known that
the maximum load of any bin is $\Theta(\log n/\log\log n)$ with high
probability (whp) \cite{Gonnet81probes}. Here, by high probability,
we mean probability $1-O(n^{-\gamma})$ for arbitrarily large $\gamma=O(1)$. An alternative variant
chooses $d$ possible bins per ball independently and uniformly at
random, and places the ball in the bin with the lowest load (breaking ties
arbitrarily). It was shown by Azar et al.~\cite{azar99lglgn} that
with this scheme the maximum load, surprisingly, drops to $\log\log
n/\log d + O(1)$ whp. When $d=2$, this is known as the power of two
choices paradigm. Applications are surveyed in \cite{Mitzenmacher01thesis,mitzenmacher01twochoice}.

Here, we are interested in applications where the two bins are picked
via hash functions $h_1$ and $h_2$, so that the two choices for a
ball, or key, can be recomputed. The most obvious such application is
a hash table with chaining. In the classic hash table by chaining (see
e.g. \cite{knuth-vol3}), keys are inserted into a table using a hash
function to decide a bin. Collisions are handled by making a
linked list of all keys in the bin. If we insert $n$ keys into a table
of size $n$ and the hash function used is perfectly random, then the
longest chain has length $\Theta(\log n/\log\log n)$ whp. If we
instead use the two-choice paradigm, placing each key in the
shortest of the two chains selected by the
hash functions,  then the maximum time to search for a key in the
two selected chains is $\Theta(\log \log n)$ whp, but this assumes that the
hash functions $h_1$ and $h_2$ are truly random, which is not realistic.
No constant time implementation of these hash functions (using space less than the size of the universe) was known
to yield maximum load $O(\log \log n)$ with high probability (see the paragraph Alternatives below).

In this paper we consider the two-choice paradigm using simple
tabulation hashing dating back to Zobrist
\cite{zobrist70hashing}. With keys from a universe
$[u]=\{0,\ldots,u-1\}$, we view a key as partitioned into $c=O(1)$
characters from the alphabet $\Sigma=[u^{1/c}]$. For each character
position $i\in [c]$, we have an independent character table $T_i$
assigning random $r$-bit hash values to the characters. The $r$-bit
hash of a key $x=(x_0,\ldots x_{c-1})$ is computed as $\bigoplus_{i\in
  [c]} T_i(x_i)$ where $\oplus$ denotes bit-wise XOR. This
takes constant time. In \cite{patrascu11charhash}, with 8-bit
characters, this was found to be as efficient as two multiplications
over the same domain, e.g., 32 or 64 bit keys.  P\v{a}tra\c{s}cu and Thorup \cite{patrascu11charhash} have
shown that simple tabulation, which is not even $4$-independent in the
classic notion of Carter and Wegman \cite{carter77universal}, has many
desirable algorithmic properties. In particular, they show that the
error probability with cuckoo hashing \cite{pagh04cuckoo} is
$O(n^{-1/3})$, and they claim (no details provided) that their analysis can be extended
to give an $O(\log\log n)$ bound on the expected maximum load in the
two-choice paradigm. For the expected bound, they can assume that
cuckoo hashing doesn't fail. However, \cite{patrascu11charhash} present
concrete inputs for which cuckoo hashing fails with
probability $\Omega(n^{-1/3})$, so this approach does not
work for high probability bounds.

\paragraph{Results}
In this paper, we show that simple tabulation works almost as well as fully random
hashing when used in the two-choice paradigm. Similar to \cite{patrascu11charhash}, we consider the bipartite case where
$h_0$ and $h_1$ hash to different tables.
Our main result is that simple tabulation gives $O(\log\log n)$ maximum load with high
probability. This is the first result giving this guarantee of $O(1)$
evaluation time for any practical hash function.

\begin{theorem}\label{thm:lglgwhp}
Let $h_0$ and $h_1$ be two independent random simple tabulation hash functions.
If $m = O(n)$ balls are put into two tables of $n$ bins sequentially using the
two-choice paradigm with $h_0$ and $h_1$, then for any constant $\gamma>0$, the
maximum load of any bin is $O(\log\log n)$ with probability $1 -
O(n^{-\gamma})$.
\end{theorem}

We also prove the following result regarding the expected maximum load,
improving on the $O(\log\log n)$ expected bound of P\v{a}tra\c{s}cu and Thorup
\cite{patrascu11charhash}.

\begin{theorem}\label{thm:lglgnexpect}
Let $h_0$ and $h_1$ be two independent random simple tabulation hash functions.
If $m = O(n)$ balls are put into two tables of $n$ bins
sequentially using the two-choice paradigm with $h_0$ and $h_1$, then the
expected maximum load is at most $\lg\lg n + O(1)$.
\end{theorem}

In contrast to the positive result of \Cref{thm:lglgwhp}, we also show that for any $k > 0$
there exists a set of $n$ keys such that the maximum load is $\ge k \lg
\lg n$ with probability $\Omega(n^{-\gamma})$ for some $\gamma >
0$. This shows that \Cref{thm:lglgwhp} is asymptotically tight and that
unlike the fully random case, $\lg\lg n + O(1)$ is not the right high
probability bound for the maximum load. 

\paragraph{Alternatives}
It is easy to see that $O(\log
 n)$-independence suffices for the classic analysis of the two-choice
 paradigm~\cite{azar99lglgn} placing $n$ balls in $\Omega(n)$
 bins. Several methods exist for computing such highly independent
 hash functions in constant time using space similar to ours
 \cite{CPT15:indep,siegel04hash,thorup13doubletab}, but these methods
 rely on a probabilistic construction of certain unbalanced constant
 degree expanders, and this only works with \emph{parameterized high
 probability in $u$}. By this we mean that to get probability $u^{-\gamma}$, we
 have to parameterize the construction of the hash function by $\gamma$, which
 results in an $\Theta(\gamma)$ evaluation time. In contrast the result of
 \Cref{thm:lglgwhp} works for any constant $\gamma$ without changing the hash
 function.
 Moreover, for $O(\log n)$-independence, the most efficient method is
 the double tabulation from \cite{thorup13doubletab}, and even if we
 are lucky that it works, it is still an order of magnitude slower
 than simple tabulation.

One could also imagine using the uniform hashing schemes from
\cite{Dahlgaard:2015,PP08}. Again, we have
the problem that the constructions only work with parameterized high
probability. The constructions from \cite{Dahlgaard:2015,PP08} are also worse
than simple tabulation in that they are more complicated and use $\Theta(n)$
space.

In a different direction, Woelfel \cite{Woelfel06} showed how to guarantee
maximum load of $\log\log n + O(1)$ using the hash family from
\cite{dietzfel03tabhash}. This is better than our result by a constant factor
and matches that of truly random hash functions. However, the result of
\cite{Woelfel06} only
works with parameterized high probability. The hash family used
in \cite{Woelfel06} is also slower and more complicated to implement than
simple tabulation. More specifically, for the most efficient
implementation of the scheme in \cite{Woelfel06}, Woelfel suggests using
the tabulation hashing from \cite{thorup12kwise} as a subroutine. However, the
tabulation hashing of \cite{thorup12kwise} is strictly more complicated than
simple tabulation, which we show here to work directly with non-parameterized
high probability.

Finally, it was recently shown by Reingold et
al.~\cite{Reingold14prgraphs} how to guarantee a maximum load of
$O(\log\log n)$ whp.~using the hash functions of
\cite{CRSW11}. These functions use a seed of $O(\log n
\log\log n)$ random bits and can be evaluated in $O((\log \log n)^2)$
time. With simple tabulation, we are not so concerned with the number
of random bits, but note that the character tables could be filled
with an $O(\log n)$-independent pseudo-random number generator. The main advantage
of simple tabulation is that we have constant evaluation time. As an
example using the result of \cite{Reingold14prgraphs} on the previously
described example of a hash table with chaining would give $O((\log\log n)^2)$
lookup time instead of $O(\log\log n)$.

\paragraph{Techniques}
In order to show \Cref{thm:lglgwhp}, we first show a structural lemma,
    showing that if a bin has high load, the underlying hash graph must either
    have a large component or large arboricity. This result is purely
    combinatorial and holds for any choice of
    the hash functions $h_0$ and $h_1$. We believe that this result is of
    independent interest, and could be useful in proving similar results for
    other hash functions. In order to show \Cref{thm:lglgnexpect}, we use
    \Cref{thm:lglgwhp} combined with
the standard idea of bounding the
probability of a large binomial tree occurring in the hash graph. An
important trick in our analysis for the expected case is to only consider a pruned binomial tree
where all degrees are large.

It remains a major open problem what happens for $m\gg n$ balls, and it does
not seem like current techniques alone generalize to this case
without the assumption that the hash functions are fully random.
We do not know of any practical hash functions that guarantee that the
difference between the maximum and the average load is $\lg \lg n + O(1)$ with
high probability when $m\gg n$---not even $\log n$-independence appears
to suffice for
this case.

\paragraph{Structure of the paper}
In \Cref{sec:prelims} we introduce well-known results and notation used
throughout the paper. In \Cref{sec:whp}, we first show that we cannot hope to
get $(1+o(1))\lg\lg n + O(1)$ maximum load whp for simple tabulation. We then
show a structural lemma regarding the arboricity of the hash graph and maximum
load of any bin. Finally, we use this lemma to prove \Cref{thm:lglgwhp}. In
\Cref{sec:expect} we prove \Cref{thm:lglgnexpect}. The proofs of
\Cref{sec:whp,sec:expect} rely heavily on a few structural lemmas regarding the
dependencies of keys with simple tabulation. The proofs of these lemmas are
included in \Cref{sec:lemmaproofs}.

\section{Preliminaries}\label{sec:prelims}

\subsection{Simple Tabulation}
Let us briefly review simple tabulation hashing. The goal is to hash keys from
some universe $[u] = \{0,\ldots, u-1\}$ into some range $\calR = [2^r]$
(i.e.~hash values are $r$ bit numbers for convenience). In tabulation hashing
we view a key $x\in [u]$ as a vector of $c>1$ \emph{characters} from the
alphabet $\Sigma = [u^{1/c}]$, i.e.~$x = (x_0,\ldots,x_{c-1})\in \Sigma^c$. We
generally assume $c$ to be a small constant.

In simple tabulation hashing we initialize $c$ independent random tables
$T_0,\ldots,T_{c-1} : \Sigma\to\calR$. The hash value $h(x)$ is then computed
as
\begin{equation}
    h(x) = \bigoplus_{i\in [c]} T_i[x_i]\enspace ,
\end{equation}
where $\oplus$ denotes the bitwise XOR operation. This is a well known scheme
dating back to Zobrist \cite{zobrist70hashing}. Simple tabulation is known to
be just $3$-independent, but it was shown in \cite{patrascu11charhash} to have
much more powerful properties than this suggests. This includes fourth moment
bounds, Chernoff bounds when distributing balls into many bins, and
random graph properties necessary for cuckoo hashing.

\paragraph{Notation}
We will now recall some of the notation used in \cite{patrascu11charhash}. Let
$S\subseteq [u]$ be a set of keys. Denote by $\pi(S,i)$ the projection of $S$
on the $i$th character, i.e.~$\pi(S,i) = \{x_i | x\in S\}$. We also use this
notation for keys, so $\pi((x_0,\ldots,x_{c-1}),i) = x_i$. A \emph{position
character} is an element of $[c]\times \Sigma$. Under this definition a key
$x\in[u]$ can be viewed as a set of $c$ position characters
$\{(0,x_0),\ldots,(c-1,x_{c-1})\}$. Furthermore we assume that $h$ is defined on
position characters as $h((i,\alpha)) = T_i[\alpha]$. This definition extends
to sets of position characters in a natural way by taking the XOR over the hash
of each position character.

\paragraph{Dependent keys}\label{sec:classify}
That simple tabulation is not $4$-independent
implies that there exists keys
$x_1,x_2,x_3,x_4$ such that for any choice of $h$, $h(x_1)$ is dependent on
$h(x_2)$, $h(x_3)$, $h(x_4)$. However, contrary to e.g.~polynomial hashing this
is not the case for all $4$-tuples. Such key dependences in simple tabulation
can be completely classified. We will state this as the following lemma, first
observed in \cite{thorup12kwise}.

\begin{lemma}[Thorup and Zhang]\label{lem:keydep}
    Let $x_1,\ldots,x_k$ be keys in $[u]^k$. If $x_1,\ldots,x_k$ are
    dependent, then there exists an $I \subseteq \{1,\ldots,k\}$ such that
    each position character of $(x_i)_{i\in I}$ appears an even number of times.

    Conversely, if each position character of $x_1,\ldots,x_k$ appears an even
    number of times, then $x_1,\ldots,x_k$ are dependent and, for any $h$,
    \[
        \bigoplus_{i=1}^k h(x_i) = 0\ .
    \]
\end{lemma}

This means that if a set of keys $(x_i)_{i\in I}$ has symmetric difference
$\emptyset$, then it is dependent. Throughout the paper, we will denote the symmetric difference between the position characters of $\{x_i\}_{i\in I}$ as
$\bigoplus_{i \in I} x_i$.

In \cite{Dahlgaard:2015}, the following lemma was shown:
\begin{lemma}[\cite{Dahlgaard:2015}]
	\label{zeroSum}
	Let $X \subset U$ be a subset with $n$ elements. The number of $2t$-tuples $(x_1, \ldots, x_{2t}) \in X^{2t}$ such that
	\[
	x_1 \oplus \cdots \oplus x_{2t} = \emptyset
	\]
	is at most $((2t-1)!!)^c n^t$. (Where $(2t-1)!! = (2t-1)(2t-3) \cdots 3 \cdot 1$)
\end{lemma}
In order to prove the main results of this paper, we prove several
similar lemmas retaining to the number of tuples with a certain number
of dependent keys.
\subsection{Two Choices}
In the two-choice paradigm, we are distributing $m$ balls (keys) into $n$ bins.
The keys arrive sequentially, and we associate with each key $x$ two random bins
$h_0(x)$ and $h_1(x)$ according to hash functions $h_0$ and $h_1$. When placing
a ball, we pick the bin with the fewest balls in it, breaking ties arbitrarily.
If $h_0$ and $h_1$ are perfectly random hash functions the maximum load of any
bin is known to be $\lg\lg n + O(1)$ whp.~if $m = O(n)$ \cite{azar99lglgn}.

\begin{definition}\label{def:cuckoo_graph}
    Given hash functions $h_0,h_1$ (as above), let the \emph{hash
    graph} 
    denote the 
    graph with bins as vertices and an edge between
    $(h_0(x),h_1(x))$ for each $x\in S$.
\end{definition}

In this paper, we assume that $h_0$ and $h_1$ map to two disjoint tables, and
the graph can thus be assumed to be bipartite.
This is a standard assumption, see e.g. \cite{patrascu11charhash},
and is actually preferable in the distributed setting.
We note that the proofs can easily be changed such that they also hold when
$h_0$ and $h_1$ map to the same table.

\begin{definition}\label{def:time}
The hash-graph $G_m$ may be decomposed into a series of nested subgraphs $G_0, G_1, \ldots, G_j, \ldots , G_m$ with edge-set $\emptyset\subset
\ldots \subset \{(h_0(x_i),h_1(x_i))\}_{i\le j}\subset \ldots \subset
\{((h_0(x_i),h_1(x_i))\}_{x_i\in S}$, which we will call the hash-graph
\emph{at the time} $0,\ldots,m$. Similarly, the load of a vertex at the time
$j$ is well-defined.
\end{definition}

\paragraph{Cuckoo hashing}
Similar to the power of $2$ choice hashing, another scheme using two hash functions is \emph{cuckoo hashing}
\cite{pagh04cuckoo}. In cuckoo hashing we wish to place each ball in one of two
random locations without collisions. This is possible if and only
if no component in the hash graph has more than one cycle.
P\v{a}tra\c{s}cu and Thorup \cite[Thm.~1.2]{patrascu11charhash} proved
that with simple tabulation, this obstruction happens with probability $O(n^{-1/3})$, and we
shall use this in our analysis of the expected maximum load for \cref{thm:lglgnexpect}.

\subsection{Graph terminology} \label{sub:binom}
The \emph{binomial tree} $B_0$ of order $0$ is a single node. The binomial
tree $B_k$ of order $k$ is a root node, which children are binomial trees
of order $B_0,\ldots,B_{k-1}$. A binomial tree of order $k$ has $2^k$ nodes
and height $k$.

The \emph{arboricity} of a graph $G$ is the minimum number of spanning forests
needed to cover all the edges of the graph. As shown by Nash-Williams~\cite{nash64decomp,nash61edge},
the arboricity of a graph $G$ equals
\[
    \max \set{\ceil{\frac{\abs{E_s}}{\abs{V_s}-1}}\mid (V_s,E_s) \textnormal{ is a subgraph of }G }
\]

\section{Maximum load with high probability}\label{sec:whp}

This section is dedicated to proving \cref{thm:lglgwhp}.
The main idea of the proof is to show that a hash graph resulting in high
maximum load must have a component which is either large or has high
arboricity. We then show that each of these cases is very unlikely.

As a negative result, we will first observe, that we cannot prove that the
maximum load is $\lg \lg n + O(1)$ or even $(1+o(1))\lg \lg n$ whp.~when
using simple tabulation.

\begin{Observation}
	Given $k = O(1)$, there exists an ordered set $S$ consisting of $n$ keys, such that when
	they are distributed into $n$ bins using hash values from simple tabulation
	the max load is
	$\ge \floor{k^{c-1}/2}\lg\lg n - O(1)$ with probability $\Omega(n^{-2(k-1)(c-1)})$.
\end{Observation}
\begin{proof}
	Consider now the set of keys $[n/k^{c-1}] \times [k]^{c-1}$ consisting of
	$n$ keys.
	For each of the positions $i = 1,\ldots,c-1$ the probability that all the
	position characters on position $i$ hash to the same value is $n^{-k+1}$.
	So with probability $n^{-(k-1)(c-1)}$ this happens for all positions
	$i = 1,\ldots,c-1$. This happens for both hash functions with probability
	$n^{-2(k-1)(c-1)}$.
	In this case $h_l(x) = h_l(x_0x_1 \ldots x_{c-1}) = h_l(x_{c-1}) \xor h_l(x_0\ldots x_{c-2})$
	is only dependent on $h_l(x_{c-1}), l \in \set{0,1}$.
	Order the keys lexicographically and insert them into the bins.
	If $n/k^{c-1} = \Omega(n)$ balls are distributed independently and uniformly
	at random to $n$ bins the maximum load would be $\lg \lg n - O(1)$ with
	probability $\Omega(1)$. (This can be proved along the lines of \cite[Thm. 3.2]{azar99lglgn}.)
	If we had exactly $2\floor{k^{c-1}/2}$ copies of $n/k^{c-1}$ independent
	and random keys the maximum load would be at least $\floor{k^{c-1}/2}$ times
	larger than if we had had $n/k^{c-1}$ independent
	and random keys. The latter is at least $\lg \lg n - O(1)$ with probability
	$\Omega(1)$.

	Since there are $k^{c-1} \ge 2\floor{k^{c-1}/2}$ copies of independent and
	uniformly random hash values we conclude that the maximum load is at
	least $\floor{k^{c-1}/2}(\lg \lg n - O(1)) = \floor{k^{c-1}/2}\lg \lg n - O(1)$
	with probability $\Omega(1)$ under the assumption that $h_l(x_0\ldots x_{c-2})$
	is constant for any $(x_0,\ldots,x_{c-2}) \in [k]^{c-1}, l \in \set{0,1}$.
	Since the latter happens with probability $n^{-2(k-1)(c-1)}$ the proof is
    finished.\qed
\end{proof}

We will now show that a series of insertions inducing a hash graph with low
arboricity and small components cannot cause a too big maximum load. Note that
this is the case for any hash functions and not just simple tabulation.

\begin{lemma}
	\label{arboricityLemma}
	Consider the process of placing some balls into bins with the two choice
	paradigm, and assume that some bin gets load $k$. Then there exists a
	connected component in the corresponding hash graph with $x$ nodes and
	arboricity $a$ such that:
	\[
	a \lg x \ge k
	\]
\end{lemma}
\begin{proof}
	Let $v$ be the node in the hash graph corresponding to the bin with load $k$.
	Let $V_k = \set{v}, E_k = \emptyset$ and define $E_l,V_l$, for
	$l = k-1, k-2,\ldots , 0$, in the
	following way: For each bin $b$ of $V_{l+1}$, add the edge corresponding to the
	$l+1$st ball landing in $b$ to the set $E_{l}$. Define $V_{l}$ to
	be the endpoints of the edges in $E_{l}$
	(see
	\cref{arboricityFigure} for a visualization). 
	\begin{figure}[htbp]
		\centering
		\includegraphics[width=.5\textwidth]{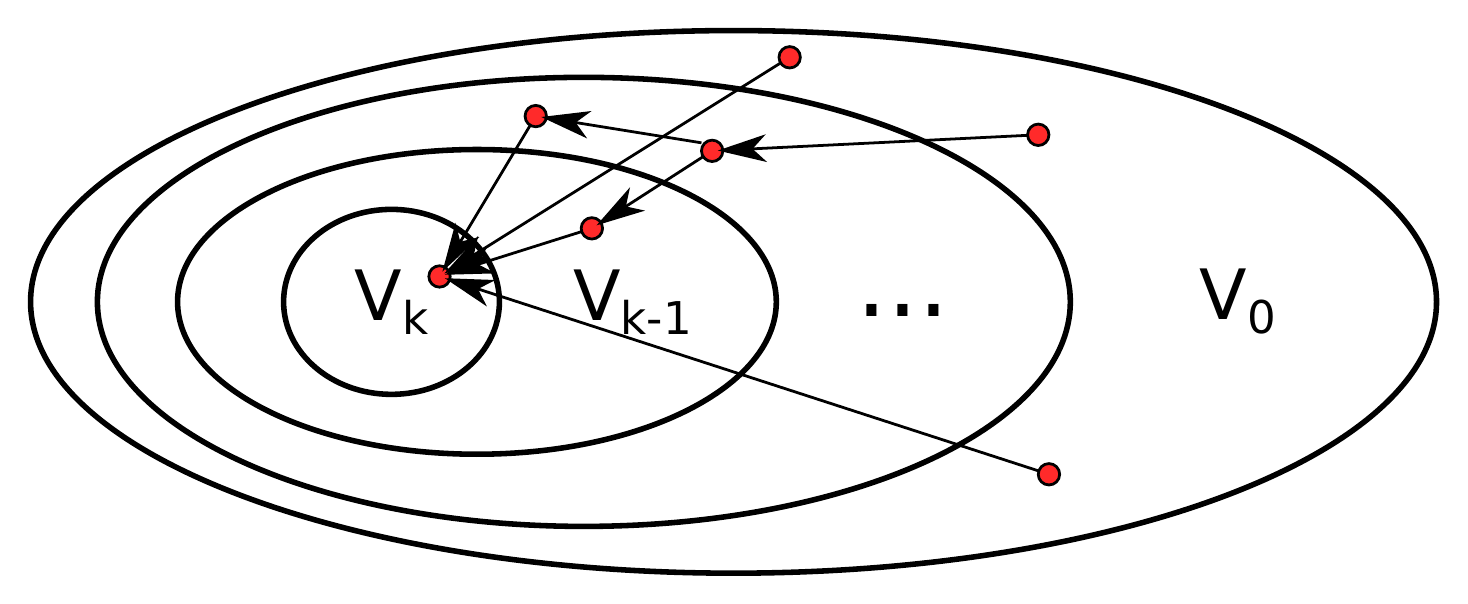}
		\caption{A visualisation of the sets $V_0,\ldots,V_k$.}
		\label{arboricityFigure}
	\end{figure}

	It is clear, that each bin of $V_l$ must have a load of at least $l$.
	Note that the definition implies that $\abs{E_l} = \abs{V_{l+1}}$ and
	$V_k \subset V_{k-1} \subset \ldots \subset V_0$. For each $l \in [k]$, let \emph{the $l$'th load-graph of $v$} denote
	the subgraph $\left ( V_l, E_l \cup E_{l+1} \cup \ldots \cup E_{k-1} \right )$.
	Let $a_l$ be defined as the following lower bound on the arboricity of
	this subgraph:
	\[
	a_l = \ceil{\frac{\abs{E_l} + \ldots + \abs{E_{k-1}}}{\abs{V_l}-1}}
	\]
	Let $a = \max_{l \in [k]} a_l$, then $a$ is a lower bound on the
	arboricity of $\left ( V_0, E_0 \cup \ldots \cup E_{k-1} \right )$. Now
	note that for each $l \in [k]$:
	\[
	\frac{\abs{E_l} + \ldots + \abs{E_{k-1}}}{\abs{V_l}-1} \le a
	\]
	Since $\abs{E_l} = \abs{V_{l+1}}$ for each $l \in [k]$ this means that:
	\[
	\abs{V_l}-1 \ge \frac{\abs{V_{l+1}} + \ldots + \abs{V_k}}{a}
	\]
	By an easy induction $\abs{V_l} \ge \left(1 + \frac{1}{a}\right)^{k-l}$,
	and therefore $\abs{V_0} \ge \left(1 + \frac{1}{a}\right)^k$. The connected
	component that contains $v$ contains at least $\abs{V_0}$ nodes, has
	arboricity $\ge a$, and:
	\[
	a \lg \abs{V_0} \ge
	a \lg \left ( 1 + \frac{1}{a} \right )^k =
	k \lg \left ( 1 + \frac{1}{a} \right )^a \ge
	k
	\]\qed
\end{proof}

Our approach is now to 1) fixing the hash graph, 2) observe that the hash
graph must have a component with a certain structure (in this case high
arboricity or big component), 3) Find a large set of independent keys, $S$ in
this component, 4) bound the probability that such a component could have
occurred using $S$.

In order to perform step 4 above we will need a way to bound the number of
sets, $S$, which have many dependent keys. This is captured by
the following lemma, which is proved in \Cref{sec:lemmaproofs}.
\begin{lemma}
	\label{whpSave}
	Let $X \subset U$ be a subset with $n$ elements and fix $k =O(1)$ and $s$
	such that $k s^{kc} < \sqrt{n}$. The number of $s$-tuples
	$(x_1,\ldots,x_s) \in X^s$ for which there exists distinct
	$y_1,\ldots,y_{(2k)^c s+1} \in X$, which are dependent on $x_1,\ldots,x_s$
	is no more than:
	\[
	n^{s-k/2} s^{O(1)}
	\]
	where the constant in the $O$-notation is dependent on $k$.
\end{lemma}

The goal is now to use the following lemma several times.
\begin{lemma}
	\label{mainWhpLemma}
	Let $X\subseteq [u]$ with $|X| = m$, and let $h_0,h_1 : [u]\to [n]$ be two
	independent simple tabulation hash functions. Fix some integer
	$k$. If $m < n/(2^8(4k)^c)$, then the maximum load of any bin when
	assigning keys using the two-choice paradigm is $O(\log\log n)$ with
	probability $1 - O(n^{-k+2})$.
\end{lemma}
\begin{proof}
	Fix the hash values of all the keys and consider the hash graph.
	Note that there
	is a one-to-one correspondence between the edges and the keys and we will
	not distinguish between the two in this proof.
	Consider any connected subgraph $C$ in the hash graph. We wish to argue
	that $C$ cannot be too big or have too high arboricity. In order to do
	this, we construct a set $S$ of independent edges contained in $C$.
	Initially let $S = \set{e}$ for some edge in $e$ in $C$.
	At all times we maintain the set $Y = Y(S)$ of keys which are dependent
	on the keys in $S$. Note that $S \subseteq Y$.
	The set $S$ is constructed iteratively in the following way:
	If there exists an edge $e\in C\sm Y$ that is incident to an edge in $S$
	add $e$ to $S$. Otherwise, if there exists
	an edge $e\in C\sm Y$, which is incident to an edge in $Y$, add $e$ to $S$.
	If neither type of edge exists we do not add more edges to $S$. Note
    that in this case C = Y.

	At any point we can partition the edges of $S$ into connected
	components $C_1,\ldots, C_t$, such that $C_1$ is the component of the
	initial edge $e$ of $S$. For each $i > 1$ we let $b_i\in Y\sm S$ be an edge
    incident to $C_i$ (such an edge must exist by the definition above).
    Order the components $C_2,\ldots,C_t$ such that $b_2 < \ldots < b_t$.
    For a visualisation of $S$ \cref{mainWhpLemmaFigure} can be consulted.
    Intuitively, since the $b_i$s cannot be chosen in too many ways,
    the ``first node'' of each $C_i$ cannot be chosen in too many ways, and
    thus it is not a problem that the set $S$ is not necessarily connected.

	\begin{figure}[htbp]
		\centering
		\includegraphics[width=.5\textwidth]{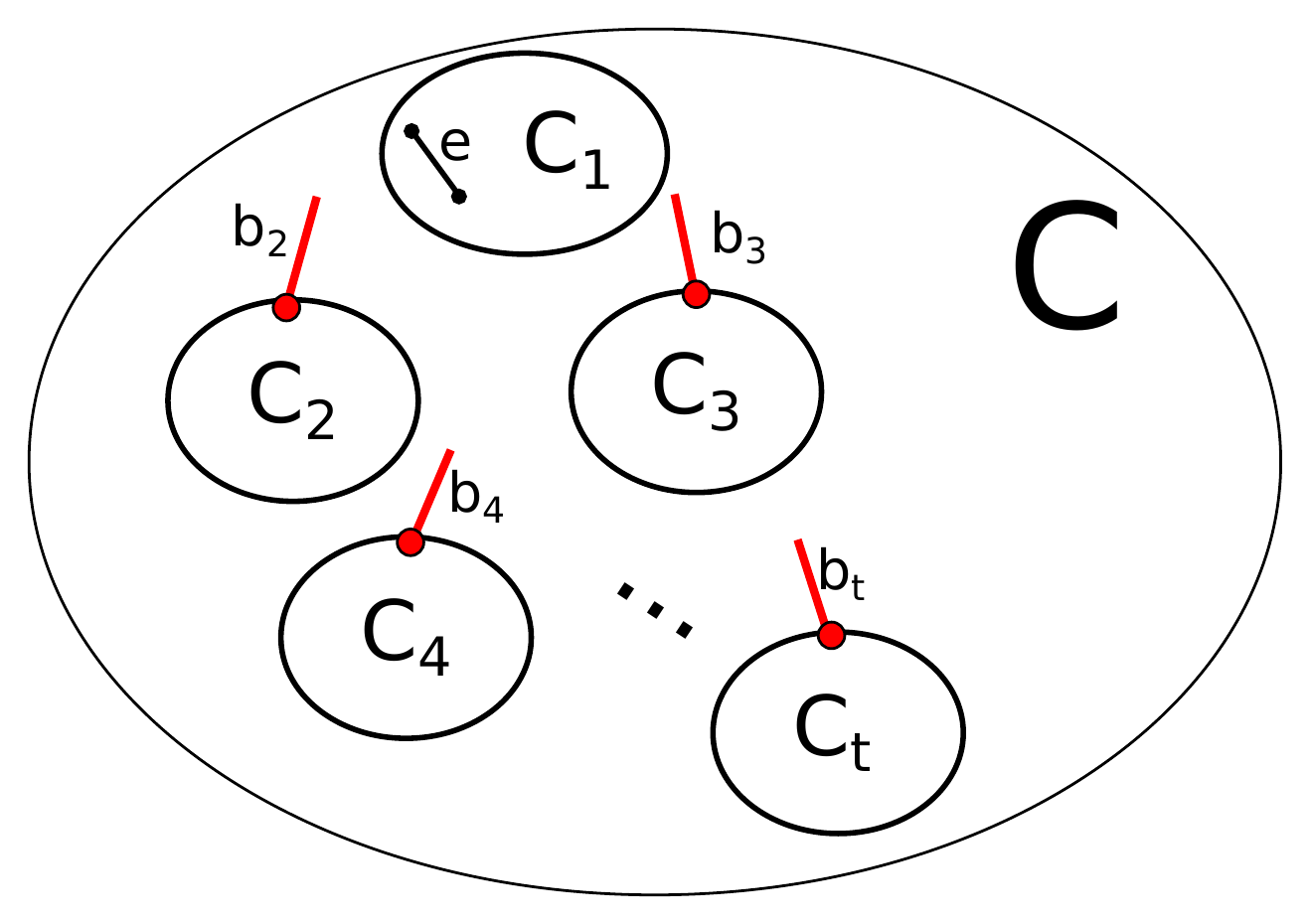}
		\caption{A visualization of the process.
			$C_1,\ldots,C_t$ correspond to components contained in $C$, and the
        red lines are the corresponding edges $b_2,\ldots, b_t \in Y$.}
		\label{mainWhpLemmaFigure}
	\end{figure}

	We stop the algorithm when either $\abs{S} \ge k \lg n$ or
	$\abs{Y} > (4k)^c \abs{S}$. We will show that the probability that this
	can happen in the hash graph is bounded by $O(n^{-k+2})$. The two cases
	are described below and the proof of each case is ended with a
	$\diamond$.

	\textbf{The algorithm stops because $\abs{Y} > (4k)^c \abs{S}$:}
	In this case we know that $|S|\le k\lg n$ since the algorithm has not
    stopped earlier and $|S|$ only grows by one in each step. Fix the size
    $\abs{S} = s$ and the number of components $t$. We wish to bound the
    number of ways $S$ could have been chosen.
    First we bound the number of ways we can choose the subgraphs
    $C_1,\ldots,C_t$ -- i.e. the edges, nodes, and keys corresponding to
    edges. Let $a_i$ be the number of nodes in the subgraph $C_i$.
	We can choose the structure of a spanning tree in each of $C_1,\ldots,C_t$
    in no more than $2^{2(a_1-1) + \ldots + 2(a_t-1)} \le 2^{2s}$ ways. Let $a
    = \sum_i a_i$ be the total number of nodes. Then this places $a-t$ of
    the edges and it remains to place $s-a+t$ edges, which can be done in at
    most $s^{2(s-a+t)}$ ways. Similarly, the number of ways that the nodes can
    be chosen is at most $n^{a-t+1}2^{2s}2^{t-1}\binom{(4k)^c s}{t-1}$ by arguing in
    the following manner: For each component $C_{i,i>1}$ we can
    describe one node by referring to $b_i$ and which endpoint the node is at
    (these are the red nodes in \Cref{mainWhpLemmaFigure}). Thus we can
    describe $t-1$ of the nodes in at most $2^{t-1}\binom{|Y'|}{t-1}$ ways, where
    $Y'$ was the set $Y$ before the addition of the last edge, so $|Y'| \le
    (4k)^cs$. These $t-1$ nodes can be picked in at most $2^{2s}$ ways, since
    there are at most $2s$ nodes in $C_1,\ldots,C_t$. The remaining $a-t+1$
    nodes can be chosen in no more than $n^{a-t+1}$ ways. Assuming that $n$ is
    larger than a constant we know by \cref{whpSave} that the number of ways to
    choose the keys in $S$ (including the order) is bounded by $s^{O(1)}
    m^{s-k}$. Hence for a fixed $a$ the total number of ways to choose $S$ is
    at most:
	\[
	2^{4s} \cdot s^{2(s-a+t)}
	\cdot n^{a-t+1} \cdot
	2^{t-1}\binom{(4k)^c s}{t-1} \cdot
	s^{O(1)} m^{s-k}
	\]
	For each of the $s$ independent keys we fix $2$ hash values, so the
	probability that those values occur is at most $n^{-2s}$. Thus the
	total probability
	that we can find such $S$ for fixed values of $s,a,t$ is at most:
	\begin{align*}
    &2^{4s} \cdot
	s^{2(s-a+t)+O(1)}
	n^{a-t+1-2s} \cdot
	2^{t-1}\binom{(4k)^c s}{t-1}
	m^{s-k}
	\\
    &\quad\le
    n
    2^{5s}
	\left ( \frac{s^2}{n} \right )^{s-a+t}
	s^{O(1)}
	\left(\frac{e(4k)^c s}{t-1}\right)^{t-1}
	\left ( \frac{m}{n} \right )^{s}
	m^{-k}
    \\
    &\quad\le
	n
	s^{O(1)}
	\left ( \frac{2^5 e (4k)^c m}{n} \right )^{s}
	m^{-k}
    \\ &\quad\le
	ns^{O(1)}
	m^{-k}
    \\ &\quad\le
	n
	(\lg n)^{O(1)}
	m^{-k}
	\end{align*}
	Since there are at at most $(2k \lg n)^3 = (\lg n)^{O(1)}$ ways to
	choose $s,a,t$ we can bound the probability by a union bound and
	get $n(\lg n)^{O(1)} m^{-k} = O(n^{-k+2})$.
	\hfill $\diamond$

	\textbf{The algorithm stops because $\abs{S} \ge k \lg n$:}
	Let $s,a,t$ have the same meaning as before. The same line of argument
    (without using \cref{whpSave}) shows that the number of ways to choose $S$
    is bounded by
	\[
	n
	s^{O(1)}
	\left ( \frac{s^2}{n} \right )^{s-a+t}
	\left ( \frac{2^5 e (4k)^c m}{n} \right )^{s}
	\le
	n
	s^{O(1)}
	2^{-s}
	\]
	Since $s = \ceil{k \lg n}$ we know that $2^{-s} \le n^{-k}$ and a union
	bound over all choices of $a,t$ suffices.
	\hfill $\diamond$

	Along the same lines we can show that $s-a+t \le k$ with probability
	$1-O(n^{-k+2})$. Here, the idea is that we need to place $s-a+t$
	additional keys when the spanning trees are fixed.
	Such a key and placement can be chosen in at most
	$s^2$ ways, but it happens with probability at most $1/n^2$ due to the
	independence of the keys.

	Now, assume there exists a component with arboricity $\alpha \ge
    2(k+2)(4k)^c$ and choose a subgraph $H$ such that $\abs{E(H)} \ge
    \alpha(\abs{V(H)}-1)$. Consider the algorithm constructing $S$ restricted
    to $H$ (and define $s$, $a$, and $t$ analogously). If the algorithm is not
    stopped early we know that $Y$ contains the edges of $H$, so $\abs{Y} \ge
    \abs{V(H)} \cdot (k+2)(4k)^c$ and thus $\abs{S} \ge (k+2)\abs{V(H)}$. This
    implies that $s-a+t \ge (k+1)\abs{V(H)} \ge k+1$, i.e.~every component has
    arboricity $\le 2(k+2)(4k)^c$ with probability $1 - O(n^{-k+2})$.

	From the analysis above we get that there exists no component with more
	than $(4k)^ck\lg n$ nodes with probability $1 - O(n^{-k+2})$. Combining
	this with \cref{arboricityLemma} we now conclude that with probability
	$1-O(n^{-k+2})$ the maximum load is upper bounded by:
	\[
	2(k+2)(4k)^c \cdot \lg \left ( (4k)^c k \lg n \right )
	=
	O \left ( \lg \lg n \right )
	\]\qed
\end{proof}

\begin{proof}[Proof of \Cref{thm:lglgwhp}]
Divide the $m=O(n)$ balls into $2^8(4\ceil{\gamma+2})^c \frac{m}{n} = O(1)$
portions of size $\le n/(2^8(4k)^c)$, apply \Cref{mainWhpLemma} to each
portion, and take a union bound.\qed
\end{proof}

\section{Bounding the expected maximum load}\label{sec:expect}

This section is dedicated to proving \cref{thm:lglgnexpect}.
The main idea is to bound the probability that a big binomial tree
appears in the hash graph. A crucial point of the proof is to consider a
subtree of the binomial tree which is chosen such that the number of leaves are
much larger than the number of internal nodes.

First of all note that by \cref{thm:lglgwhp}, the probability that the maximum load
is more than $k_0 \cdot \lg \lg n$ is $O(n^{-1})$ for some constant $k_0 > 1$.
Hence it suffices to prove that the probability that the maximum load is larger
than $\lg \lg n + r+1$ is at most $O((\lg \lg n)^{-1})$ for some constant $r$
depending on $m/n$ and $c$.
\begin{Observation}\label{obs:Bk}
		If there exists a bin with load at least
		$k+1$ then either there is a component with more edges than nodes, or the
		binomial tree $B_k$ is a subgraph of the hash graph.
	\end{Observation}
	\begin{proof}
		Assume no component has more edges than nodes. Then,
		removing at most one edge from each component yields a forest.
		One edge per component will
		at most increase the load by $1$, so consider the remaining forest.

		Consider now the order in which the keys are inserted, and use induction on
		this order. Define $G_j$ to be the graph after the $j$th key is inserted.
		The induction hypothesis is that if a bin has load $k$, then it is the root in
		a subtree which is $B_k$. For $G_0$ it is easy to see. Consider now the
		addition of the $j$th key and assume that the hypothesis holds. Assume that the
		added key corresponds to the edge $(u,v)$ and that the load of bin $u$
		increases to $l$. Since there are no cycles, node $G_j$ must have edges
		$(u,v_0),\ldots, (u,v_{l-1})$, and by the induction hypothesis $v_0,\ldots,v_{l-1}$ are
		roots of disjoint binomial trees $B_0,\ldots,B_{l-1}$, so $u$ is the root of a
		$B_l$.\qed
	\end{proof}


Let \emph{double cycle} denote any of the minimal obstructions described in \cite{patrascu11charhash}, that is, either a cycle and a path between two vertices of the cycle, or two cycles linked by a path. Note that any connected graph with two cycles (not necessarily disjoint) contains a double cycle as a subgraph.

\begin{Observation}\label{obs:smalldoublecycle}
	If there exists a bin with load at least $k+1$, then either the binomial tree $B_k$ is a subgraph of the hash graph, or a double cycle with at most $4k+4$ edges is a subgraph of the hash graph.
\end{Observation}
\begin{proof}
	As in the proof of \Cref{arboricityLemma}, let $v$ be the node with load $\ge k+1$, and let $(V_0,E_0\cup E_1\cup\ldots\cup E_{k+1})$ be the $k+1$st load-graph of $v$.

	If the $k+1$st load-graph of $v$ has no more edges than vertices, it must contain $B_k$ as a subgraph, as noted in \Cref{obs:Bk}. Otherwise, take $v$ as root and 
	consider a breadth first spanning tree, $T$. It must have height at most $k+1$, and there must be two edges $(u,w),(y,z)$ of the combined load-graph not in $T$.
	Furthermore, these edges cannot have both endpoints have maximal distance $k+1$ from $v$. Thus, the union $T\cup\{(u,w),(y,z)\}$ has at most $4(k+1)$ edges and must contain a double cycle as a subgraph.
	\qed
\end{proof}

We are now ready to prove \cref{thm:lglgnexpect}.
\begin{proof}[Proof of \cref{thm:lglgnexpect}.]
    if $m(1+\eps) < n$ we know from \cite[Thm.~1.2]{patrascu11charhash}
	that no component of the hash graph
	contains a double cycle with probability $O(n^{-1/3})$. Looking into the proof
	we see that there exists no double cycle consisting of at most
	$s$ edges with probability $(O(m/n))^{s}n^{-1/3}$ even when $m > n$. In the
	terminology of \cite{patrascu11charhash}, $\lg(n/m)$ bits per edge is saved in
	the encoding of the hash-values. But when $\lg(n/m) < 0$ we add $\lg(m/n)$ bits
	to the encoding instead. If the double cycle consist of $s$ edges this is $s\lg(m/n)$
	extra bits in
	the encoding, i.e. that the bound on the probability is multiplied with
	$(O(m/n))^s$. This means that we only need to bound the probability that there
	exists a binomial tree $B_k, k = \ceil{\lg \lg n + r}$, because, according to \Cref{obs:smalldoublecycle}, any bin with
	load $k+1$ will either imply the existence of $B_k$ in the hash graph or the
    existence of a double cycle consisting of $4k+4 = O(\lg \lg
    n)$ edges, and
	the latter happens with probability $(\lg n)^{O(1)} n^{-1/3} = O(n^{-1/4})$.

	Say that the hash graph contains a binomial tree $B_k$. Consider the subtree
	$T_{k,d}$ defined by removing the children of all nodes that have less than
	$d$ children, where $d \le k$ is some constant to be defined (see \cref{fig:binom_s}). Note that
    $T_{k,d}$ has $(d+1)2^{k-d}-1$ edges. We will now follow the same
    approach as in the proof of \Cref{thm:lglgwhp}, in the sense that we
    construct as set $S$ of independent keys from $T_{k,d}$, and show that
    this is unlikely.
	We construct the ordered set $S$ by traversing $T_{k,d}$ in the following way:
	Order the edges in increasing distance from the root and on each level from left to right.
	Traverse the edges in this order. A given edge is added to the ordered set $S$ if the following
	two requirements are fulfilled:
	\begin{itemize}
		\item   After the edge is added $S$ corresponds to a connected subgraph of $T_{k,d}$.
		\item   The key corresponding to the edge is independent of all the keys corresponding
		to the edges in $S$.
	\end{itemize}
	A visualization of the set $S$ can be seen in \cref{fig:binom_s}.
\begin{figure}[htbp]
    \centering
    \includegraphics[width=.7\textwidth]{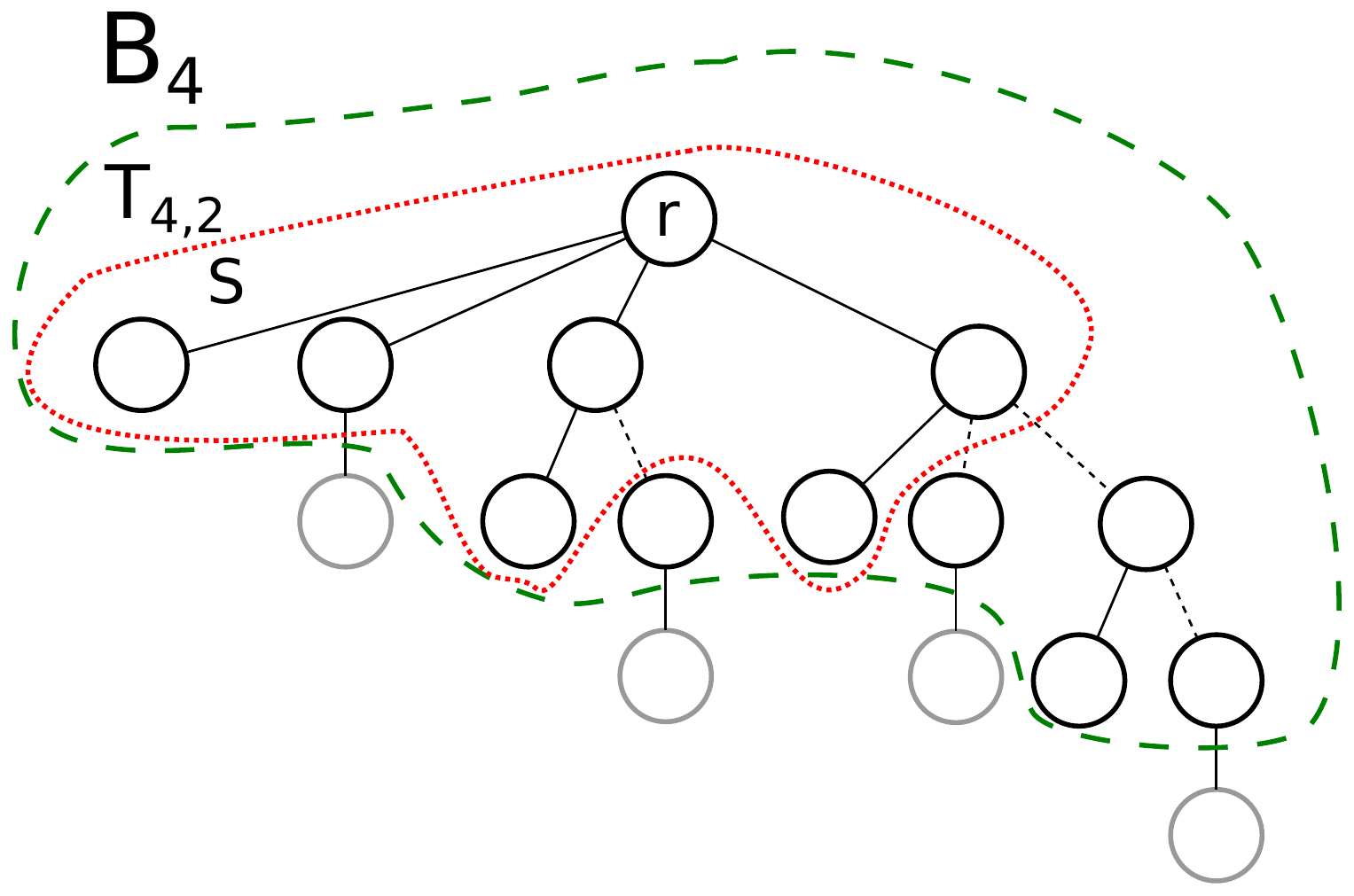}
    \caption{Example of $T_{4,2}$ and the corresponding set $S$. The dashed
        edges correspond to key dependencies at the time the edge is considered
    in the order. This example would correspond to case 2.}
    \label{fig:binom_s}
\end{figure}

	We will think of $S$ as a set of edges, but also as a set of independent keys.
The idea is to bound the probability that we could find such a set $S$. We will split the proof into four cases depending on $S$, and each will end with a $\Diamond$.

\textbf{Case 1: $s := \abs{S} = (d+1)2^{k-d}-1$:}
In this case every edge of the tree is independent, and
there are at most $m^s$ different ways to choose
the ordered set $S$. Note that there are $2^{k-d}$ groups of $d$ leaves which
have the same parent. The set $S$ corresponds to the same subgraph of the hash
graph regardless of the ordering of these leaves. Since we only want to bound
the probability that we can find such $S$, we can thus chose the edges of
$S$ in at most
$m^s \left ( \frac{1}{d!} \right )^{2^{k-d}}$ ways.
For a given choice of $S$ there are $s-1$ equations
$h_k(x) = h_k(y)$ which must be fulfilled where $k \in \set{1,2}$ and $x,y$ are keys
in $S$. Since the keys in $S$ are independent, the probability that this happens for a given
$S$ is at most $2n^{-(s-1)}$. By a union bound on all the choices of $S$
the probability that such an $S$ exists is at most:
\begin{align*}
m^s \left ( \frac{1}{d!} \right )^{2^{k-d}} (2n^{-(s-1)})
& \le
\frac{2m^s}{n^{s-1}} \left ( \frac{1}{\sqrt[d+1]{d!}} \right )^{2^{k-d}(d+1)}
\\
&
\le
2n \cdot \left ( \frac{m}{n\sqrt[d+1]{d!}} \right )^s
\end{align*}
We now pick $d$ and $r$ such that
$\frac{m}{n\sqrt[d+1]{d!}} < \frac{1}{2}$ and $r \ge d+1$. It then follows that
$s \ge 2 \lg n$ and the probability is bounded by $2n^{-1}$.
\hfill $\Diamond$

For case 2 and case 3, we will use the following lemma, which is proved in
\cref{sec:lemmaproofs}.

\begin{lemma}
	\label{lem:stdSave}
	Let $X \subset U$ be a subset with $n$ elements and fix $s$ such that $s^c \le \frac{4}{5}n$. The number of $s$-tuples $(x_1,\ldots,x_s), x_i \in X$ for which there exists $y \in X, y \neq x_1,\ldots, x_s$ such that $h(y)$ is dependent of $h(x_1), \ldots, h(x_s)$ is no more than:
	\[
	s^4 \frac{3^c}{6} n^{s-1}
	\le
	s^{O(1)}
	n^{s-1}
	\]
\end{lemma}

\textbf{Case 2: All the edges incident to the root lie in $S$:}
Let $S'$ be defined in a similar manner as $S$:
Order the edges in increasing distance from the root and on each level
from left to right as before.
Traverse the edges in this order, and add the edges to $S'$ if the corresponding
key is independent of the keys in $S'$. However, stop this traversal the first time
a dependent key occurs. In this way $S'$ will be an ordered subset of $S$
and the tree-structure will only depend on $s' = \abs{S'}$. Fix this value $s'$.
Since there is a key which is dependent on the keys in $S'$ there are at most
$s'^{O(1)}m^{s'-1}$ ways to choose $S'$ by \Cref{lem:stdSave} assuming that
$s'^c \le \frac{4}{5}m$, i.e.
assuming that $n$ is larger than some constant depending on $c$.

Every internal node of $T_{k,d}$ has exactly $d$ children that are
leaves. Therefore, there can be at most one node in $S'$ having less than
$d$ children that are leaves and belong to $S'$.
Let $v_1,\ldots,v_l$ denote the internal nodes in $S'$, where $l$ is
the number of internal nodes. Let $w_i$ denote the number of children of $v_i$
that are leaves.
Similar to case 1, the structure of $S'$ is independent of the order of the
leaves with the same parent. Therefore $S'$ can be chosen in at most
$s'^{O(1)}m^{s'-1}\prod_{i=1}^l \frac{1}{w_i!}$ ways. Since
$w_i! \ge \left ( \frac{w_i}{e} \right )^{w_i}$ we see that:
\[
\prod_{i=1}^l \frac{1}{w_i!} \le
\prod_{i=1}^l \left ( \frac{e}{w_i} \right )^{w_i}
\]
Letting $w = \sum_{i=1}^l w_i$ the concavity of $x \to x\log(e/x)$
combined with Jensen's inequality yields:
\[
\prod_{i=1}^l \left ( \frac{e}{w_i} \right )^{w_i}
\le
\left ( \frac{le}{w} \right )^{w}
\]
At most one of the $w_i$'s can be smaller than $d$, so
wlog.~assume that $w_1,\ldots,w_{l-1} \ge d$. The total number of nodes must
be at least $l + d(l-1)$, i.e. $s' \ge l + d(l-1)$ giving $l \le \frac{s'+d}{d+1}$.
Since $l+w = s'$ we see that:
\[
\frac{l}{w} \le
\frac{\frac{s'+d}{d+1}}{\frac{s'd-d}{d+1}} =
\frac{1}{d} \cdot \frac{s'+d}{s'-1} \le
\frac{2}{d}
\]
Where the last inequality holds assuming that $n$
(and hence $s' \ge \lg \lg n$) is larger
than a constant. Since $w \ge (s-1)\frac{d}{d+1}$ we see that:
\[
\prod_{i=1}^l \frac{1}{w_i!} \le
\left ( \left ( \frac{2e}{d} \right )^{\frac{d}{d+1}} \right )^{s'-1}
\]
Assume that $d$ is chosen such that
$\left ( \frac{2e}{d} \right )^{\frac{d}{d+1}} \le \frac{n}{2m}$. The number of
cases that we need to consider is then at most:
\[
s'^{O(1)}m^{s'-1}\prod_{i=1}^l \frac{1}{w_i!}
\le
s'^{O(1)}\left ( \frac{n}{2} \right )^{s'-1}
\]

Since $S'$ is a tree
there are $s'-1$ equalities on the form $h_k(x) = h_k(y)$ where
$k \in \set{1,2}, x,y \in S'$ that must be satisfied if $S'$ occurs. Since we
know the tree structure from knowing $s'$ there are at most two ways
two choose these equalities. This means that the probability that a specific
$S'$ occurs is bounded by $2n^{-(s'-1)}$.
For a fixed $\abs{S'} = s'$ the probability that there
exists $S'$ with $s'$ elements is therefore bounded by:
\[
2s'^{O(1)}\left ( \frac{n}{2} \right )^{s'-1}n^{-(s'-1)}
=
2s'^{O(1)} 2^{-s'+1}
\]
A union bound over all $s' \ge \lg \lg n$ now yields the desired upper bound:
\begin{align*}
    &\sum_{s' \ge \ceil{\lg \lg n}}
    2s'^{O(1)} 2^{-s'+1}
    \\ &\qquad\le
    2^{-\lg \lg n+3}
    \sum_{k \ge 1}
    \left (k+\ceil{\lg \lg n}-1 \right )^{O(1)}
    2^{-k}
    \\ &\qquad\le
    \frac{8}{\lg n}
    \ceil{\lg \lg n}^{O(1)}
    \sum_{k \ge 1}
    k^{O(1)}
    2^{-k}
    \\ &\qquad=
    \frac{(\lg \lg n)^{O(1)}}{\lg n}
\end{align*}
\hfill $\Diamond$

\textbf{Case 3: Not all, but at least $(\lg \lg n)/2$ edges incident to the root lie in $S$:}
Let $S' \subset S$ be the set of independent keys adjacent to the root, and set
$s' = |S'|$. By \Cref{lem:stdSave}, $S'$
can be chosen in no more than $\frac{s'^{O(1)}m^{s'-1}}{s'!}$ ways since there must exist a key
(corresponding to an edge incident to the root) which is dependent on the keys
in $S'$ and the order of the keys are irrelevant. Since all the keys in $S'$ are
independent, the probability that $h_0(x)$ or $h_1(x)$ are the same for all the keys $x \in S'$
is at most $2n^{-(s'-1)}$. So the probability that such a $S'$ can be found is at most:
\begin{align*}
\frac{s'^{O(1)}m^{s'-1}}{s'!} \cdot (2n^{-(s'-1)})
&=
2s'^{O(1)}
\frac{\left (
	\frac{m}{n}
	\right )^{s'-1}}{s'!}
    \\ &\le
2s'^{O(1)}
\left (
\frac{me}{ns'}
\right )^{s'-1}
\\ &=
O((\lg \lg n)^{-1})
\end{align*}
\hfill $\Diamond$

For case 4, we will use the following generalization of
\Cref{lem:stdSave}, which is proved in \cref{sec:lemmaproofs}.
\begin{lemma}\label{lem:extraSave}
	Let $X\subseteq U$ with $|X| = n$ and fix $s$ such that $s^c\le
	\frac{4}{5}n$. The number of $s$-tuples $(x_1,\ldots, x_s)$ for which
	there
	exists distinct $y_1,\ldots, y_k\in X\sm\{x_1,\ldots,x_s\}$ for
	$k\ge\max(s-1,5)$ such that each
	$h(y_i)$ is dependent on $h(x_1),\ldots, h(x_s)$ is at most
	\[
	s^6\frac{15^c}{120}n^{s-2} + s^6\frac{9^c}{36}n^{s-2} +
	s^5\frac{9^c}{4}n^{s-3/2} = s^{O(1)}n^{s-3/2}\ .
	\]
\end{lemma}

\textbf{Case 4: There are less than $(\lg \lg n)/2$ edges incident to the root in $S$:}
Let $S' \subset S$ be the set of keys corresponding to the edges from $S$
incident to the root and let $s' = \abs{S'}$.
Since the other keys incident to the root must be
dependent on the keys from $S'$, \Cref{lem:extraSave} states that $S'$ can
be chosen in at most $s'^{O(1)}m^{s'-3/2}$ ways. Since all the keys in $S'$ are
independent the probability that $h_0(x)$ or $h_1(x)$ are the same for
all the keys $x \in S'$ is at most $2n^{-(s'-1)}$. Thus, the probability of such a
set $S'$ occurring is bounded by:
\begin{align*}
s'^{O(1)}m^{s'-3/2}
\cdot
( 2n^{-(s'-1)} )
&\le
s'^{O(1)} 2n^{-1/2}
\left ( \frac{m}{n} \right )^{s'-3/2}
\\ &=
(\log n)^{O(1)}n^{-1/2}
\end{align*}
\hfill $\Diamond$

This covers all cases for the set $S$.\qed
\end{proof}

Consider the case of distributing $m$ balls into $n$ bins. Note that the proof actually gives an expected maximum load of
$O(m/n) + \lg \lg n + O(1)$ if
$m/n = o((\lg n)/(\lg \lg n))$. However, this only matches the behaviour
of truly random hash functions under the assumption that $m = O(n)$.

The same techniques can be used to show that
$\Omega\left(\frac{m}{n}\log n\right)$-independent hash functions yield a
maximum load of
$O(m/n) + \lg \lg n + O(1)$ with high probability (this is essentially case 1 in the proof).
This implies that $\Omega(\lg n)$-independence hashing is sufficient to give the
same theoretical guarantees as truly random hash functions in the context
of the power of two choices when $m = O(n)$.

\section{Proofs of structural lemmas}\label{sec:lemmaproofs}
In this section we prove the lemmas used in \Cref{sec:whp,sec:expect}.

The following lemma is a generalization of \Cref{zeroSum} and is proved in
\cite{Dahlgaard:2015}.
\begin{lemma}[\cite{Dahlgaard:2015}]
	\label{zeroSumProvable}
	Let $A_1,\ldots,A_{2t} \subset U$ be subsets of $U$. The number of $2t$-tuples $(x_1, \ldots, x_{2t}) \in A_1 \times \cdots \times A_{2t}$ such that
	\begin{align}
	\label{eq:zeroSumProvable}
	x_1 \oplus \cdots \oplus x_{2t} = \emptyset
	\end{align}
	is at most $((2t-1)!!)^c \prod_{i=1}^{2t} \sqrt{\abs{A_i}}$. (Where $(2t-1)!! = (2t-1)(2t-3) \cdots 3 \cdot 1$)
\end{lemma}

We use \Cref{zeroSumProvable} in our proof of \Cref{lem:extraSave}:

\begin{proof}[Proof of \Cref{lem:extraSave}]
	For each $j = 1,\ldots,k$ let $I_j \subset \oneToN{s}$ be such that
    $y_j = \bigoplus_{i \in I_j} x_i$ for all choices of $h$. There are at
    most $s\choose |I_j|$ ways to choose $I_j$. Note that
    $(x_i)_{i\in\{1,\ldots,s\}\sm I_j}$ can be chosen in at most
    $n^{s-|I_j|}$ ways and by \Cref{zeroSum} $(x_i)_{i\in I_j}$ can be
    chosen in at most $((|I_j|)!!)^cn^{(|I_j|+1)/2}$ ways. I.e. for a fixed
    value of $|I_j|$ an upper bound is:
    \begin{equation}\label{eq:setlimit}
        {s\choose |I_j|}((|I_j|)!!)^cn^{s-|I_j|/2+1/2}\ .
    \end{equation}
    If $|I_j| > 3$ we can use \eqref{eq:setlimit} to get an upper bound on the
    number of such $s$-tuples of

	\[
	s \cdot
	\binom{s}{5}
	(5!!)^c
	n^{s-5/2+1/2}
	\le
	s^6 \frac{15^c}{120} n^{s-2}
	\]
	Now assume that $\abs{I_j} = 3$ for $j = 1,\ldots,k$. Note that the sets $I_j$ must be distinct
	and since $I_j \subset \oneToN{s}$ and $k \ge \max\set{5,s-1}$ there must exist $j,l \in \oneToN{k}$
	such that $\abs{I_j \cap I_l} \le 1$.

	Case $\abs{I_j \cap I_l} = 0$: In this case the number of possible values for $(x_i)_{i \in I_j}$,
	$(x_i)_{i \in I_l}$, $I_j$, and $I_l$ is, by \Cref{zeroSum}, no more than:
	\[
	\binom{s}{3,3}
	\left (
	(3!!)^c
	n^{2}
	\right )^2
	\]
	and the remaining $x_i$'s can be chosen in at most $n^{s-6}$ ways giving an upper bound of:
	\[
	\frac{s^6}{(3!)^2}
	\left (
	(3!!)^c
	n^{2}
	\right )^2
	n^{s-6}
	=
	s^6
	\frac{9^c}{36}
	n^{s-2}
	\]

	Case $\abs{I_j \cap I_l} = 1$: $I_j$ and $I_l$ can be chosen in $\binom{s}{3,2} \cdot 3$ ways.
	By \Cref{zeroSum} $(x_i)_{i \in I_j}$ can be chosen in $(3!!)^c n^{2}$ ways. The number of ways
	to choose $(x_i)_{i \in I_l}$ once $(x_i)_{i \in I_j}$ is then by \Cref{zeroSumProvable}
	no more than $(3!!)^c n^{3/2}$ since we choose one of the $A_i$'s to be a singleton. The
	remaining $x_i$'s can be chosen in at most $n^{s-5}$ ways giving a total upper bound of:
	\[
	\binom{s}{3,2} \cdot 3
	\cdot
	(3^c n^2) \cdot
	(3^c n^{3/2}) \cdot
	n^{s-5}
	\le
	s^5
	\frac{9^c}{4}
	n^{s-3/2}
	\]
	Which concludes the proof.\qed
\end{proof}

The proof of \Cref{lem:stdSave} follows the same argument.
\begin{proof}[Proof of \Cref{lem:stdSave}]
	Since $h(y)$ is dependent of $h(x_1), \ldots, h(x_s)$ there exists a subset $I \subset \{1,\ldots,s\}$ such that for all choices of $h$:
	\[
	\bigoplus_{i \in I} x_i = y
	\]
	Fix $\abs{I}$ and note that $\abs{I} \ge 3$ (by 3-independence). There are $\binom{s}{\abs{I}}$ ways to choose $I$. Note that $(x_i)_{i \in \oneToN{s} \backslash I}$ can be chosen in at most $n^{s - \abs{I}}$ ways and by \Cref{zeroSum} $(x_i)_{i \in I}$ can be chosen in at most $((\abs{I})!!)^c n^{(\abs{I}+1)/2}$ ways. I.e. for a fixed value of $\abs{I}$ an upper bound is:
	\[
	\binom{s}{\abs{I}}
	((\abs{I})!!)^c
	n^{s-\abs{I}/2+1/2}
	\]
	We can show that this upper bound is maximal when $\abs{I} = 3$. Since $\abs{I}$ is odd it suffices to show that the value decreases when $\abs{I}$ increases by $2$ as long as $\abs{I}+2 \le s$.
	Consider the following fraction:
	\begin{align*}
    &\frac{
		\binom{s}{\abs{I}}
		((\abs{I})!!)^c
		n^{s-\abs{I}/2+1/2}
	}{
	\binom{s}{\abs{I}+2}
	((\abs{I}+2)!!)^c
	n^{s-(\abs{I}+2)/2+1/2}
}
\\&\qquad=
\frac{
	(\abs{I}+1)(\abs{I}+2)
	n
}{
(s-\abs{I})(s-\abs{I}-1)
(\abs{I}+2)^c
}
\\&\qquad\ge
\frac{\frac{4}{5}n}{s^c}
\end{align*}
By the assumption this fraction is at least $1$, and hence the upper bound
decreases with $\abs{I}$. Therefore, as $|I|$ grows there are fewer ways to
describe $(x_1,\ldots,x_s)$.

Since $3\le |I| \le s$, the number of ways to choose
$(x_1,\ldots,x_s)$ is upper bounded by:
\[
s \cdot
\binom{s}{3}
(3!!)^c
n^{s-3/2+1/2}
\le
s^4 \frac{3^c}{6} n^{s-1}
\]\qed
\end{proof}

In order to prove \Cref{whpSave} we will need the following combinatorial
lemma.

\begin{lemma}
	\label{whpSaveCombi}
	Let $s,k,c \ge 1$ be integers and $A_1,\ldots,A_{(2k)^c s+1}$ be
	non-empty subsets of $\oneToN{s}$, such that for every $B \subset \oneToN{s}$:
	\[
	\abs{\set{A_i \mid A_i \subset B}} \le \abs{B}^c
	\]
	Then there exists $I \subset \oneToN{(2k)^c s+1}$ such that $\abs{I} \le k$ and
	\[
	f(I) \stackrel{def} = \abs{\bigcup_{i \in I} A_i} - \abs{I} \ge k
	\]
\end{lemma}
\begin{proof}
	Let $I \subset \oneToN{(2k)^c s+1}$ be such that $\abs{\cup_{i \in I} A_i}
	< 2k$.
	We want to show
	that there exists $J = I \cup \set{r}$ for some $r \in \oneToN{(2k)^c s+1}$ such
	that $f(J) > f(I)$.
	Let $A = \cup_{i \in I} A_i$ and assume for the sake of contradiction
	that no such $r$ exists. This implies that $\abs{A_r \setminus A} \le 1$
	for all $r \in \oneToN{(2k)^c s + 1}$. I.e. that each $A_r$ is contained in
	one of the sets
	\[
	\left ( A \cup \set{1} \right ),
	\left ( A \cup \set{2} \right ),
	\ldots,
	\left ( A \cup \set{s} \right )
	\]
	By assumption, each of these sets contains no more than $(|A|+1)^c$
	sets $A_r$, and thus they contain at most $(|A|+1)^c s$ sets combined. This
	means that
	\[
	(2k)^c s + 1 \le \left(\abs{A}+1\right)^c s \le (2k)^c s\ ,
	\]
	which is a contradiction. Thus there must exists an $r$ such that
	$f(I\cup\set{r}) > f(I)$.

	Now consider the following greedy algorithm: Let $I := \emptyset$ and
	iteratively set $I := I \cup \set{r}$ for such an $r$ until
	$\abs{\cup_{i \in I} A_i} \ge 2k$.
	Since $f(I)$ increases in each step, the algorithm stops after at most $k$
	steps. This implies that $f(I) \ge 2k-k=k$ and $\abs{I} \le k$ as
    desired.\qed
\end{proof}

We can use \Cref{whpSaveCombi} to show \Cref{whpSave}.

\begin{proof}[Proof of \Cref{whpSave}]
	For each $i \in \oneToN{(2k)^c s + 1}$ let $A_i \subset \oneToN{s}$ be such
	that:
	\[
	\bigoplus_{j \in A_i} x_j = y_i
	\]
	By \cref{whpSaveCombi} there exists $I \subset \oneToN{(2k)^c + 1}$ such
	that for $A := \cup_{i \in I} A_i$, $\abs{A} - \abs{I} \ge k$, $\abs{I} \le
	k$.

	It is enough to show the lemma for a fixed $\abs{A}$ and $\abs{I}$ {as
		these can be chosen in at most $ks = O(s)$.} Fix $\abs{A} = a$
	and $\abs{I} = r$.

	Let $I = \set{v_1,\ldots,v_r}$ and for each $j \in \oneToN{r}$ define
	$B_j$ as:
	\[
	B_j = A_{v_j} \setminus \left ( \bigcup_{i < j} A_{v_i} \right )
	\]
	Wlog.~assume that $a = \sum_{j < r} \abs{B_j} \le 2k$. (Otherwise there
	exists a smaller set $I$)
	The number of ways to choose $(B_j)_{1 \le j \le r}$ is at most
	$\binom{s}{a} r^a$: There are $\binom{s}{a}$ ways to choose $A$ and
	$r^a$ ways to partition $A$ into $B_1,\ldots,B_r$.

	Now, fix the choice of $B_1,\ldots,B_r$. We will bound the number of ways to
	choose $(x_i)_{i \in B_j}$ given that $(x_i)_{i \in B_1}$, \ldots,
	$(x_i)_{i \in B_{j-1}}$ are chosen. The number of ways to choose $A_j$ is
	at most $2^{2k}$ for $j\in I$. For a fixed choice of $A_j$ the number of ways to choose
	$(x_i)_{i \in B_j}$ is at most $(\abs{A_j}!!)^c n^{(\abs{B_j}+1)/2}$
	by \cref{zeroSumProvable}. Hence, the number of ways to choose $(x_i)_{i
		\in A}$ is at most:
    \begin{align*}
	\prod_{j=1}^r \left ( 2^{2k} (\abs{A_j}!!)^c n^{(\abs{B_j}+1)/2} \right )
    &\le
	2^{2kr} (a!!)^{rc} n^{(a+r)/2}
    \\&\le
	2^{2k^2} (a!)^{kc} n^{(a+r)/2}
    \end{align*}
	The number of ways to choose the remaining $(x_i)_{i \notin A}$
	is trivially
	bounded by $n^{s-a}$ giving a total upper bound on the number of ways to
	choose $(x_i)_{i \in \oneToN{s}}$ of:
	\[
	\binom{s}{a} k^a
	2^{2k^2} (a!)^{kc} n^{s-a/2+r/2}
	\]
	Now note that if $a < s$:
    \begin{align*}
    &\frac{\binom{s}{a+1} k^{a+1} 2^{2k^2} ((a+1)!)^{kc} n^{s-(a+1)/2+r/2}}
	{\binom{s}{a} k^a 2^{2k^2} (a!)^{kc} n^{s-a/2+r/2}}
    \\&\qquad=
	\frac{(s-a)k(a+1)^{kc}}
	{(a+1)n^{1/2}}
    \\&\qquad<
	1
    \end{align*}
	This implies that the upper bound is biggest when $a$ is smallest, i.e.
	when $a = r+k$. In this case the upper bound is:
    \begin{align*}
    &\binom{s}{k} k^k
	2^{2k^2} ((r+k)!)^{kc} n^{s-k/2}
    \\&\qquad\le
	\binom{s}{k} k^k 2^{2k^2} ((2k)!)^{kc} n^{s-k/2}
    \\&\qquad=
	s^{O(1)} n^{s-k/2}
    \end{align*}
	which concludes the proof.\qed
\end{proof}

\bibliographystyle{amsplain}
\bibliography{general}

\end{document}